\documentclass[12pt, a4paper, reqno]{amsart}
\usepackage{mathrsfs}
\usepackage{bbm}
\usepackage{pgf,tikz}
\usepackage{tabu}
\usepackage{amsmath,amscd,amssymb,latexsym}
\usepackage[all]{xy}

\input xypic

\evensidemargin=\oddsidemargin
\DeclareMathVersion{can}
\DeclareMathAlphabet{\can}{OT1}{cmss}{m}{n}
\newtheorem{thm}{Theorem}[section]

\newtheorem{lem}[thm]{Lemma}

\newtheorem{rem}[thm]{Remark}
\newtheorem{exa}[thm]{Example}
\theoremstyle{definition}

\theoremstyle{fact}

\theoremstyle{conjecture}

\numberwithin{equation}{section}

\begin{document}
\title[Optimal minimal Linear codes from posets]
{Optimal minimal Linear codes from posets}

\author[J. Y. Hyun]{Jong Yoon Hyun}
 \address{\rm Konkuk University, Glocal Campus, 268 Chungwon-daero Chungju-si Chungcheongbuk-do 27478, South Korea}
\email{hyun33@kku.ac.kr}

\author[H. K. Kim]{Hyun Kwang Kim}
 \address{\rm Pohang University of Science and Technology, 77 Cheongam-ro, Nam-Gu, Pohang, 37673, South Korea}
 \email{hkkim@postech.ac.kr}

\author[Y. Wu]{ Yansheng Wu}
\address{\rm Department of Mathematics, Ewha Womans University, 
52, Ewhayeodae-gil, Seodaemun-gu, Seoul, 03760, South Korea}
\email{wysasd@163.com}

\author[Q. Yue]{Qin Yue}
\address{\rm Department of Mathematics, Nanjing University of Aeronautics and Astronautics, Nanjing, Jiangsu, 211100, P. R. China; State Key Laboratory of Cryptology, P. O. Box 5159, Beijing, 100878, P. R. China}
\email{yueqin@nuaa.edu.cn}

\subjclass[2010]{94B05,   06A11, 06E30}
\keywords{poset, optimal binary linear code, minimal binary linear code, weight distribution}

\date{\today}


\baselineskip=20pt

\begin{abstract}   

Recently, some infinite families of minimal and optimal binary linear codes were constructed from simplicial complexes by Hyun {\em et al.} We extend this construction method to arbitrary posets. Especially, anti-chains are corresponded to simplicial complexes. 

In this paper, we present two constructions of binary linear codes from hierarchical posets of two levels. In particular, we determine the weight distributions of binary linear codes associated with hierarchical posets with two levels. Based on these results, we also obtain some optimal and minimal binary linear codes not satisfying the condition of Ashikhmin-Barg. 

\end{abstract}

\maketitle

\bigskip
\section{Introduction}

Let $\Bbb F_2$ be the finite field with order two. For positive integers $n, k$ and $d$,
an $[n, k, d]$ binary linear code $\mathcal{C}$ is a $k$-dimensional subspace of $\mathbb{ F}_2^n$ with minimum (Hamming) distance $d$. We sometimes denote by $w_{\min}$ instead of $d$. 
The support $\mathrm{supp}(v)$ of a vector $v \in \Bbb F_2^n$ is defined by the set of nonzero coordinate positions. The Hamming weight $wt(v)$ of $v\in \mathbb{F}^n_2$ is defined by the size of $\mathrm{supp}(v)$.

We say that a linear code is {\em distance-optimal} if it has the highest minimum distance with prescribed length and dimension. An $ [n,k,d]$ linear code is called {\em almost distance-optimal} if the code $[n, k, d + 1]$ is optimal, see \cite[Chapter 2]{HP}. For an $[n, k, d]$ binary linear code, the Griesmer bound (see \cite{G}) states that
\begin{eqnarray*}
n\ge \sum_{i=0}^{k-1}\bigg\lceil {\frac{d}{2^i}}\bigg \rceil,  
\end{eqnarray*} 
where $\lceil {x} \rceil$ denotes the smallest integer greater than or equal to $x$. 
We say that a linear code is a {\em Griesmer code} if it meets the Griesmer bound with equality. One can verify that Griesmer codes are distance-optimal.

Let $A_i$ be the number of codewords in a linear code $\mathcal C$
 with Hamming weight $i$. The {\em weight enumerator} of 
 $\mathcal C$ is defined by
 $1+A_1z+A_2z^2+\cdots+A_nz^n.$
The sequence $(1, A_1, A_2, \ldots, A_n)$ is called the {\em weight distribution} of  
 $\mathcal C$. A code $\mathcal{C}$ is $t$-weight if the number of nonzero $A_{i}$ in the sequence $(A_1, A_2, \ldots, A_n)$ is equal to $t$. The study of the weight distribution of a linear code is important in both
 theory and application because the weight distribution of a linear code can be used to estimate the error correcting capability and
 the error probability of error detection and correction with respect to some algorithms.


Constructing minimal linear codes is an active research topic because they could be decoded with the minimum distance decoding method \cite{AB}, and have applications in secret sharing and secure two-party computation \cite{CDY,CCP,CMP,DY,KY,M,W1, W2, W3,W4, YD}.
Aschikhmin and Barg \cite{AB} presented a sufficient condition for a linear code to be minimal. The first example of a minimal linear code violating Ashikhmin-Barg's condition was given by  Cohen {\em et al.} in \cite[Remark 1]{CMP}. Chang {\em et al.} \cite{CH} presented one infinite family of minimal binary linear codes violating Ashikhmin-Barg's condition. Ding, Heng and Zhou \cite{DHZ, HDZ} presented a necessary and sufficient condition for $q$-ary linear codes to be minimal, and using this characterization they obtained some infinite families of minimal binary and ternary linear codes not satisfying the condition of Ashikhmin-Barg. For more general case, Bartoli and Bonini \cite{BB} found one infinite families of minimal $q$-ary linear codes for which the Ashikhmin-Barg's condition does not hold, that is, they generalized the constructions of Ding, Heng, Zhou \cite{HDZ} to any field $\Bbb F_q$ with odd order.

In this paper, we focus on constructions of distance-optimal binary (minimal) linear codes by using posets.  Since the order ideals of hierarchical posets are easy to handle, i.e., they are just
a disjoint union of subsets, we determine parameters of codes generated by one or two order ideals in hierarchical posets with two levels. In Section 2, we introduce basic concepts on posets and some known results on minimal linear codes. In Section 3, we present the closed form of the generating function associated with an order ideal which allows us to compute efficiently the Hamming weights of linear codes. In Section 4, we introduce hierarchical posets with two levels and determine the form of order ideals. 
In Section 5, we determine the weight distributions of binary linear codes associated with hierarchical posets with two levels.
In Section 6, we derive some optimal and minimal binary linear codes based on the results of Section 5. Finally we conclude this paper in Section 7.





For convenience of the reader, we list the following notations used in this paper:

\begin{tabular}{ll}

$[n]$& the set $\{1,2,\ldots,n\}$,\\
$\mathbb{P}=([n],\preceq)$ & a partially ordered set on $[n]$, \\
$I$& an order ideal of $\mathbb{P}$,\\
$\mathcal{O}_{\mathbb{P}}$& the set of all order ideals of $\mathbb{P}$,\\
$I(\mathbb{P})$ & the set of order ideals of $\mathbb{P}$ that are contained in $I$,\\
$\mathcal{I}$ & a set of order ideals of $\mathbb{P}$, i.e., $\mathcal{I}\subseteq\mathcal{O}_{\mathbb{P}}$,\\
$\mathcal{I}(\mathbb{P})$ & the set of order ideals of $\mathbb{P}$ that are contained some order ideals in $\mathcal{I}$,\\
$a-A$& the set $\{a-b:b\in A\}$,\\
$A\backslash B$& the set $\{x: x\in A\mbox{ and } x\notin B\}$,\\
$|A|$& the number of elements of a set $A$,\\
$D^{c}$& the complement of a subset $D$ of $[n]$.\\
\end{tabular}

\section{Preliminaries}

\subsection{Posets}$~$

We say that $\mathbb{P} = ([n],\preceq)$ is a {\bf partially ordered set} (abbreviated as a poset) if $\mathbb{P}$ is a partial order relation on $[n]$, that is, for all $i,j,k \in [n]$ we have that: $(i)$ $i\preceq i$; $(ii)$ $i\preceq j$ and $j\preceq i$ imply $i=j$; $(iii)$ $i\preceq j$ and $j\preceq k$ imply $i\preceq k$. 

Let $\mathbb{P}=([n],\preceq)$ be a poset. Two distinct elements $i$ and $j$ in $[n]$ are called {\bf comparable} if either $i\preceq j$ or $j\preceq i$, and incomparable otherwise. It is said that a poset $\mathbb{P}$ is an {\bf anti-chain} if every pair of distinct elements is incomparable.

A nonempty subset $I$ of $\mathbb{P}$ is called an {\bf order ideal} if $j \in I $ and $i \preceq j$ imply $i \in I$. For a subset $E$ of $\mathbb{P}$, the smallest order ideal of $\mathbb{P}$ containing $E$ is denoted by $\langle E\rangle$. For an order ideal $I$ of $\mathbb{P}$, we use $I(\mathbb{P})$ to denote the set of order ideals of $\mathbb{P}$ which is contained in $I$.
Let $\mathcal{I}=\{I_1,\ldots, I_m\}$ be a subset of $\mathcal{O}_{\mathbb{P}}$.
We define  
\begin{equation}
\mathcal{I}(\mathbb{P})=\{J\in \mathcal{O}_{\mathbb{P}}: J\subseteq I\in \mathcal{I}\}=\bigcup_{i=1}^m I_i(\mathbb{P}).
\end{equation}
Then $\mathcal{I}(\mathbb{P})$ is an order ideal of $\mathcal{O}_{\mathbb{P}}$ with partial order $\subseteq$.

 \begin{rem}{\rm
 We point out that if $\mathbb{P}$ is an anti-chain, then $\mathcal{I}(\mathbb{P})$ is a simplicial complex. In \cite{CH} and \cite{HLL}, the authors  produced infinite families of distance-optimal (minimal) linear codes from simplicial complexes. In the last two sections, we will employ hierarchical posets (see Section 4) to derive infinite families of distance-optimal (minimal) linear codes.   
 }
 \end{rem}

\begin{exa}{\rm
Let $\mathbb{P}=([4],\preceq)$ be a poset with $1\prec 2$, $3\prec 4$ and the other pairs $(i,j)$ are incomparable. Let $\mathcal {I}_i$ be subsets of $\mathcal{O}_{\mathbb{P}}$ for $i=1,2,3$.  

$(1)$ If $\mathcal {I}_1=\langle \{2\}\rangle=\{\{1,2\}\}$, then $\mathcal{I}_1(\mathbb{P})=\{\emptyset, \{1\}, \{1,2\}\}.$

$(2)$ If $\mathcal {I}_2=\{\{1,2\}, \{3,4\}\}$, then $\mathcal {I}_2(\mathbb{P})=\{\emptyset, \{1\}, \{1,2\}, \{3\}, \{3,4\}\}.$

$(3)$ If $\mathcal {I}_3=\{\{1,2\}, \{1,3,4\}\}$, then $\mathcal {I}_3(\mathbb{P})=\{\emptyset, \{1\}, \{1,2\}, \{3\}, \{3,4\}, \{1,3,4\}\}.$
}
\end{exa}


\subsection{Minimal linear codes}$~$

For two vectors $u,v\in \Bbb F_2^n$, we say that $u$ covers $v$ if $\mathrm{supp}(v)\subseteq\mathrm{supp}(u)$. 
A nonzero codeword $u$ in a linear code $\mathcal{C}$ is said to be {\bf minimal} if $u$ covers the zero vector and the $u$ itself but no other codewords in the code $\mathcal{C}$. A linear code $\mathcal{C}$ is said to be {\bf minimal} if every nonzero codeword in the code $\mathcal{C}$ is minimal.

The following lemma developed by Aschikhmin and Barg \cite{AB}   is a useful criterion for a linear code to be minimal. 

\begin{lem} 
A linear code $\mathcal{C}$ over $\Bbb F_2$ with minimum distance $w_{\min}$ is minimal provided that $w_{\min}/{w_{\max}}>1/2$, where $w_{\max}$ denotes the maximum nonzero Hamming weight in the code $\mathcal{C}$.
 \end{lem}
 
The following lemma is useful in finding a minimal linear code violating the condition of Aschikhmin-Barg. 
\begin{lem}\cite[Theorem 3.2]{DHZ} 
Let $\mathcal{C}$ be a linear code over $\mathbb{F}_2$. Then the code  $\mathcal{C}$ is minimal if and only if $wt(a+b)\neq wt(a)-wt(b)$ for each pair of distinct nonzero codewords $a$ and $b$ in the code $\mathcal{C}$.
\end{lem}

\section{Generating functions for order ideals of $\mathcal{O}_{\mathbb{P}}$}


There is a bijection between $\mathbb{F}_2^n$ and $2^{[n]}$ being the power set of $[n]$, defined by $v\mapsto$ supp$(v)$. {\bf Throughout this paper, we will identify a vector in $\mathbb{F}_2^n$  with its support. }

Let $X$ be a subset of $\mathbb{F}_2^n$.
Define 
$$\mathcal{H}_{X}(x_1,x_2\ldots, x_n)=\sum_{u\in X}\prod_{i=1}^nx_i^{u_i}\in \mathbb{Z}[x_1,x_2, \ldots, x_n],
$$
where $u=(u_1,u_2,\ldots, u_n)\in \mathbb{F}_2^n$ and $\mathbb{Z}$ is the ring of integers. We observe that

$(1)$ $\mathcal{H}_{\emptyset}(x_1,x_2\ldots, x_n)=0,$

$(2)$ $\mathcal{H}_{X}(x_1,x_2\ldots, x_n)+\mathcal{H}_{X^{c}}(x_1,x_2\ldots, x_n)=\mathcal{H}_{\mathbb{F}_2^n}(x_1,x_2\ldots, x_n)=\prod_{i\in[n]}(1+x_i)$.\\

We now present the closed form of the generating function associated with an order ideal of $\mathcal{O}_{\mathbb{P}}$.
It allows us to compute efficiently the Hamming weights of linear codes defined in Section 5.

\begin{thm} \label{th1}
Let $\mathbb{P}=([n],\preceq)$ be a poset and let $\mathcal{I}=\{I_1,\ldots, I_k\}$ be a subset of $\mathcal{O}_{\mathbb{P}}$.
Then 
 \begin{align}
 \mathcal{H}_{\mathcal{I}(\mathbb{P})}(x_1,x_2\ldots, x_n)=\sum_{\emptyset\neq S\subseteq \mathcal{I}}(-1)^{|S|+1}\mathcal{H}_{\bigcap_{I\in S}I(\mathbb{P})}(x_1,x_2\ldots, x_n),
 \end{align}
In particular, we have that $|\mathcal{I}(\mathbb{P})|=\sum_{\emptyset\neq S\subseteq \mathcal{I}}(-1)^{|S|+1}|{\bigcap_{I\in S}I(\mathbb{P})}|$.
\end{thm}

\begin{proof} 
By the inclusion-exclusion principle, $$\bigcup_{j=1}^tA_j=\sum_{k=1}^t(-1)^{k+1}\sum_{1\le i_1<i_2<\cdots<i_k\le t} A_{i_1}\cap \cdots \cap A_{i_k},$$ where $A_1,\ldots,A_t$ are subsets of $[n]$. 
Let $1_X$ stand for the indicator function of a subset $X$ of $[n]$, i.e., $1_X(u)=1$ if and only if $u\in X$.
Then \begin{eqnarray*}&&\mathcal{H}_{\mathcal{I}(\mathbb{P})}(x_1,x_2\ldots, x_n)
\\
&=&\sum_{u\in \mathcal{I}(\mathbb{P})}\prod_{i=1}^nx_i^{u_i}= \sum_{u\in \mathcal{I}(\mathbb{P})}1_{\mathcal{I}(\mathbb{P})}(u)\prod_{i=1}^nx_i^{u_i}\\
&=& \sum_{u\in \mathcal{I}(\mathbb{P})}1_{\bigcup_{i=1}^m {I_i}(\mathbb{P})}(u)\prod_{i=1}^nx_i^{u_i}\\
&=& \sum_{u\in \mathcal{I}(\mathbb{P})}\sum_{k=1}^{m}(-1)^{k+1}\sum_{1\le i_1<i_2<\cdots<i_k\le m} 1_{{I_1}(\mathbb{P})\bigcap\cdots \bigcap {I_k}(\mathbb{P})}(u)
\prod_{i=1}^nx_i^{u_i}\\
&=& \sum_{k=1}^{m}(-1)^{k+1}\sum_{1\le i_1<i_2<\cdots<i_k\le m} \sum_{u\in  {I_1}(\mathbb{P})\bigcap\cdots \bigcap {I_k}(\mathbb{P})}
\prod_{i=1}^nx_i^{u_i}\\
&=&\sum_{\emptyset\neq S\subseteq \mathcal{I}}(-1)^{|S|+1}\mathcal{H}_{\bigcap_{I\in S}I(\mathbb{P})}(x_1,x_2\ldots, x_n).
\end{eqnarray*}

This completes the proof.
\end{proof}

\begin{exa} {\rm
Let $\mathbb{P}_1=([4],\preceq)$ be a poset given by the Hasse diagram in Figure 1. Let $\mathcal{I}=\{I_1,I_2\}$ be a subset of $\mathcal{O}_{\mathbb{P}_1}$, where $I_1=\{1,2\},I_2=\{3,4\}$. Then

$(1) $  $I_1(\mathbb{P})=\{\emptyset, \{2\},\{1,2\}\}$;
 
$(2)$ $I_2(\mathbb{P})=\{\emptyset, \{4\},\{3,4\}\}$;

$(3)$ $\mathcal{I}(\mathbb{P})=I_1(\mathbb{P})\cup I_2(\mathbb{P})=\{\emptyset, \{2\},\{1,2\}, \{4\},\{3,4\}\}$;

$(4)$ $\mathcal{H}_{I_1(\mathbb{P})}(x_1,x_2, x_3, x_4)=1+x_2+x_1x_2$;
$\mathcal{H}_{I_2(\mathbb{P})}(x_1,x_2, x_3, x_4)=1+x_4+x_3x_4$; $\mathcal{H}_{\mathcal{I}(\mathbb{P})}(x_1,x_2, x_3, x_4)=1+x_2+x_1x_2+x_4+x_3x_4$.
Since $I_1(\mathbb{P})\bigcap I_2(\mathbb{P})=\{\emptyset\}$, we can confirm Eq. (3.1) in Theorem \ref{th1}. 
}
\end{exa}

\begin{center}
\begin{tikzpicture}
\draw (2, -5)[fill = black] circle (0.05);
\draw (0, -5)[fill = black] circle (0.05);
\draw (0, -7)[fill = black] circle (0.05);
\draw (2, -7)[fill = black] circle (0.05);

\draw (7, -5)[fill = black] circle (0.05);
\draw (5, -5)[fill = black] circle (0.05);
\draw (6, -6)[fill = black] circle (0.05);
\draw (6, -7)[fill = black] circle (0.05);

\draw[thin] (0, -7) --(0,-5);
\draw[thin] (2, -7) --(2,-5);

\draw[thin] (6, -6) --(5,-5);
\draw[thin] (7, -5) --(6,-6);
\draw[thin] (6, -6) --(6,-7);
\node  at (0, -5) [left] {$~~~1~~$};
\node  at (0, -7) [left] {$~~2~~$};
\node  at (2, -7) [right] {$~~4~~$};
\node  at (2, -5) [right] {$~~3~~$};

\node  at (5, -5) [left] {$~~~4~~$};
\node  at (6, -6) [left] {$~~1~~$};
\node  at (6, -7) [right] {$~~2~~$};
\node  at (7, -5) [right] {$~~3~~$};

\node  at (1, -8) [below] {Figure 1 ~$\mathbb{P}_1=([4], \preceq$)};

\node  at (6, -8) [below] {Figure 2 ~$\mathbb{P}_2=([4], \preceq$)};
\end{tikzpicture}
\end{center}

\begin{exa} {\rm
Let $\mathbb{P}_1=([4],\preceq)$ be a poset given by the Hasse diagram in Figure 2. 
Let $\mathcal{I}=\{I_1,I_2\}$, where $I_1=\{1,2,3\},I_2=\{1,2,4\}$. Then

 $(1) $  $I_1(\mathbb{P})=\{\emptyset, \{2\},\{1,2\}, \{1,2,3\}\}$;
 
 $(2)$ $I_2(\mathbb{P})=\{\emptyset, \{2\},\{1,2\},\{1,2,4\}\}$;

$(3) $ $\mathcal{I}(\mathbb{P})=\{\emptyset, \{2\},\{1,2\}, \{1,2,3\},\{1,2,4\}\}$;

$(4)$  $\mathcal{H}_{\mathcal{I}(\mathbb{P})}(x_1,x_2, x_3, x_4)=1+x_2+x_1x_2+x_1x_2x_3+x_1x_2x_4$; $\mathcal{H}_{I_1(\mathbb{P})}(x_1,x_2, x_3, x_4)=1+x_2+x_1x_2+x_1x_2x_3$;
$\mathcal{H}_{I_2(\mathbb{P})}(x_1,x_2, x_3, x_4)=1+x_2+x_1x_2+x_1x_2x_4$.
Since $I_1(\mathbb{P})\bigcap I_2(\mathbb{P})=\{\emptyset,\{2\}, \{1,2\}\}$, we can confirm Eq. (3.1) in Theorem \ref{th1}. 
}
\end{exa}

\section{Hierarchical posets with two levels}

Let $m$ and $n$ be positive integers with $m\leq n$. 
We say that $\mathbb{H}(m,n)=([n],\preceq)$ is a hierarchical poset with two levels if 
$[n]$ is the disjoint union of two incomparable subsets $U=\{1,\ldots,m\}$ and $V=\{m+1,\ldots,n\}$, and $i\prec j$ whenever $i\in U$ and $j\in V$. Its Hasse diagram is given in Figure 3. By convention, $\mathbb{H}(m,m)$ is considered as an anti-chain.

\begin{center}
\begin{tikzpicture}
\draw (-3, -5)[fill = black] circle (0.05);
\draw (-1, -5)[fill = black] circle (0.05);

\draw (1, -5)[fill = black] circle (0.03);
\draw (1.5, -5)[fill = black] circle (0.03);
\draw (2, -5)[fill = black] circle (0.03);

\draw (6, -5)[fill = black] circle (0.05);
\draw (4, -5)[fill = black] circle (0.05);

\draw (-2, -8)[fill = black] circle (0.05);
\draw (0, -8)[fill = black] circle (0.05);

\draw (1, -8)[fill = black] circle (0.03);
\draw (1.5, -8)[fill = black] circle (0.03);
\draw (2, -8)[fill = black] circle (0.03);

\draw (3, -8)[fill = black] circle (0.05);
\draw (5, -8)[fill = black] circle (0.05);

\draw[thin] (-3, -5) --(-2, -8)-- (-1, -5)--(0, -8);
\draw[thin] (-3, -5) --(0, -8);
\draw[thin] (-3, -5) -- (3, -8)-- (-1, -5);
\draw[thin] (-3, -5) --(5, -8)-- (-1, -5);

\draw[thin] (4, -5) --(-2, -8);
\draw[thin] (4, -5) --(5, -8);
\draw[thin] (4, -5) --(0, -8);
\draw[thin] (4, -5) --(3, -8);

\draw[thin] (6, -5) --(-2, -8);
\draw[thin] (6, -5) --(5, -8);
\draw[thin] (6, -5) --(0, -8);
\draw[thin] (6, -5) --(3, -8);

\node  at (-3, -4.5) [above] {$~~~~m+1~~~~$};
\node  at (-1, -4.5) [above] {$~~~~m+2~~~~$};
\node  at (4, -4.5) [above] {$~~~~n-1~~~~$};
\node  at (6, -4.5) [above] {$~~~~n~~~~$};

\node  at (-2, -9) [above] {$~~~~1~~~~$};
\node  at (0, -9) [above] {$~~~~2~~~~$};
\node  at (3, -9) [above] {$~~~~m-1~~~~$};
\node  at (5, -9) [above] {$~~~~m~~~~$};

\node  at (2, -10) [below] {Figure 3~~ $\mathbb{H}(m,n)$};
\end{tikzpicture}
\end{center}

\begin{lem} \label{lem1}
Every order ideal of $\mathbb{H}(m,n)$ can be expressed by $A\cup B$ for $A\subseteq [m]$, $B\subseteq [n]\setminus[m]$, and one of the following holds: $(i)$ $B=\emptyset$; $(ii)$ $ B \neq  \emptyset $ and $A=[m]$.
\begin{proof}
The proof is straightforward.
\end{proof}







\end{lem}

\begin{lem}\label{lem2} 
Let $I=A\cup B$ for $A\subseteq [m]$, $B\subseteq [n]\setminus[m]$ be an order ideal of $\mathbb{P}=\mathbb{H}(m,n)$.  

$(1)$ If $B=\emptyset$, then 
\[\mathcal{H}_{I(\mathbb{P})}(x_1,x_2\ldots, x_n)=\sum_{u\in I(\mathbb{P})}\prod_{i=1}^nx_i^{u_i}=\prod_{i\in A}(1+x_i).
\]
In particular, we have that $|I(\mathbb{P})|=2^{|A|}$.

$(2)$ If $B\neq \emptyset$, then 
\[\mathcal{H}_{I(\mathbb{P})}(x_1,x_2\ldots, x_n)=\prod_{i\in[m]}(1+x_i)+\prod_{i\in[m]}x_i(\prod_{j\in B}(1+x_j)-1).
\]
In particular, we have that
$
|I(\mathbb{P})|=2^m+2^{|B|}-1.
$
\end{lem}

\begin{proof}
The proof is straightforward.
\end{proof}

\begin{exa} 
{\rm
Let us consider the hierarchical poset $\mathbb{P}=\mathbb{H}(2,4)$ with two levels. 

$(1)$ If $I=\{1,2\}$, then $\mathcal{H}_{I(\mathbb{P})}(x_1,x_2\ldots, x_4)=1+x_1+x_2+x_1x_2=(1+x_1)(1+x_2).$

$(2)$ If $I=\{1,2,3\}$, then $\mathcal{H}_{I(\mathbb{P})}(x_1,x_2\ldots, x_4)=1+x_1+x_2+x_1x_2+x_1x_2x_3=(1+x_1)(1+x_2)+x_1x_2x_3.$

$(3)$ If $I=\{1,2,3,4\}$, then $\mathcal{H}_{I(\mathbb{P})}(x_1,x_2\ldots, x_4)=1+x_1+x_2+x_1x_2+x_1x_2x_3+x_1x_2x_4+x_1x_2x_3x_4=(1+x_1)(1+x_2)+x_1x_2[(1+x_3)(1+x_4)-1].$
}
\end{exa}

\section{Weight distributions of binary linear codes}

In this section, we determine the weight distributions of two types of linear codes defined in (5.1) and (5.5) below which are involved with hierarchical posets of two levels.

\subsection{Linear codes from hierarchical posets with two levels}$~$



Let $\mathbb{P}$ be a poset on $[n]$ and $D=(\mathcal{I}(\mathbb{P}))^c$ considered as the complement of  $\mathcal{I}(\mathbb{P})$ in $2^{[n]}$, where $\mathcal{I}=\{I_1,\ldots,I_k\}\subseteq\mathcal{O}_{\mathbb{P}}$. Recall that there is a bijection between $\mathbb{F}_2^n$ and $2^{[n]}$ being the power set of $[n]$, defined by $v\mapsto$ supp$(v)$. If we identify a vector in $\mathbb{F}_2^n$  with its support, then here $D$ can be viewed as a subset of $\mathbb{F}^n_2$. 
We define a linear code as follows:
 \begin{equation}
 \mathcal{C}_{D}=\{c_{D, u}=(u\cdot x)_{x\in D}: u\in \Bbb F_2^n\},
 \end{equation} where $\cdot$ denotes Euclidean inner product of two elements in
$\Bbb F_2^n$. Then the length of the code $\mathcal{C}_{D}$ is $|D|$ and its dimension is at most $n$.
The Hamming weight of the codeword $c_{D, u}$ of the code $\mathcal{C}_{D}$ becomes that \begin{eqnarray}wt(c_{D, u})&=&|D|-\frac12\sum_{y\in\Bbb F_2}\sum_{x\in D}(-1)^{(u\cdot x)y}\nonumber\\
&=&\frac {|D|}2-\frac12\sum_{x\in D}(-1)^{u\cdot x}\nonumber\\
&=&\frac {|D|}2+\frac12\sum_{x\in \mathcal{I}(\mathbb{P})}(-1)^{u_1 x_1}(-1)^{u_2 x_2}\cdots(-1)^{u_n x_n}\nonumber\\
&=&\frac {|D|}2+\frac12\mathcal{H}_{\mathcal{I}(\mathbb{P})}((-1)^{u_1},(-1)^{u_2},\ldots, (-1)^{u_n}),
\end{eqnarray}
where the second equation holds due to the fact that we just take $y=0$ and $y=1$.






By Lemmas 4.1 and 4.2, the order ideals and their generator functions are easy to determined. This will bring us great convenience in computation of the weight distribution of $\mathcal{C}_{D}$. In this subsection, we always assume that  $\mathbb{P}$ is a hierarchical poset $\mathbb{H}(m,n)$ with two levels, which was introduced in Section 4. 

The weight distribution of the code $\mathcal{C}_{D}$ generated by one or two order ideals will be determined. To do so, for $u=(u_1,u_2,\ldots, u_n)\in \mathbb{F}_2^n$, write $u=(v,w)$, where $v=(u_1,\ldots, u_m)$ and $w=(u_{m+1},\ldots, u_n)$. 
For $X$ a subset of $\mathbb{F}_2^n$, we use $\chi(u|X)$ to denote a Boolean function in $n$-variable, and $\chi(u|X)=1$ if and only if $u\bigcap X=\emptyset$.

\begin{thm} 
Let $\mathbb{H}(m,n)$ be a hierarchical poset with two levels and $I=A\cup B$ an order ideal of $\mathbb{H}(m,n)$ for $A\subseteq [m]$, $B\subseteq [n]\setminus[m]$. Set $\mathcal{I}=\{I\}$.  

$(1)$ If $B=\emptyset$, then the length of the code $\mathcal{C}_{D}$ is $2^n-2^{|A|}$ and its weight distribution is given in Table 1.
\begin{table}[h]  
\caption{Theorem 5.1 (1)}   
\begin{tabu} to 0.4\textwidth{X[2,c]|X[2,c]}  
\hline 
\rm{Weight}&\rm{Frequency}\\ 
\hline
$0$&$1$\\ 
\hline
$2^{n-1}$&$2^{n-|A|}-1$\\
\hline
$2^{n-1}-2^{|A|-1}$& $2^n-2^{n-|A|}$ \\ 
\hline
\end{tabu}  
\end{table}

$(2)$ If $B\neq \emptyset$, then the length of the code $\mathcal{C}_{D}$ is $2^n-2^m-2^{|B|}+1$ and its weight distribution is given in Table 2. 
\begin{table} [h]
\caption{Theorem 5.1 (2)}   
\begin{tabu} to 0.8\textwidth{X[2,c]|X[2,c]}  
\hline 
\rm{Weight}&\rm{Frequency}\\
\hline
$0$&$1$\\
\hline
$2^{n-1}$&$2^{n-m-|B|}-1$\\
\hline
$2^{n-1}-2^{|B|-1}$& $2^{n-m}-2^{n-m-|B|}$ \\
\hline
$2^{n-1}+1-2^{m-1}-2^{|B|}$&$2^{n-1-|B|}$ \\
\hline
$2^{n-1}+1-2^{m-1}-2^{|B|-1}$&$2^{n-1}-2^{n-1-|B|}$ \\
\hline
$2^{n-1}-2^{m-1}$&$2^{n-1-|B|}-2^{n-m-|B|}$ \\
\hline
$2^{n-1}-2^{m-1}-2^{|B|-1}$&$2^{n-1}-2^{n-1-|B|}-2^{n-m}+2^{n-m-|B|}$ \\
\hline
\end{tabu}   
\end{table} 
\end{thm}
\begin{proof} 

Let $\mathbb{P}=\mathbb{H}(m,n)$.
Recall that for $X$ a subset of $\mathbb{F}_2^n$, we use $\chi(u|X)$ to denote a Boolean function in $n$-variable, and $\chi(u|X)=1$ if and only if $u\bigcap X=\emptyset$.
We also recall that for $u=(u_1,u_2,\ldots, u_n)\in \mathbb{F}_2^n$, write $u=(v,w)$, where $v=(u_1,\ldots, u_m)$ and $w=(u_{m+1},\ldots, u_n)$. 

$(1)$ Let $B=\emptyset$. By Lemma \ref{lem2}, the length of the code $\mathcal{C}_{D}$ is $2^n-|\mathcal{I}(\mathbb{P})|=2^n-2^{|A|}$ and  
\begin{multline}
\mathcal{H}_{\mathcal{I}(\mathbb{P})}((-1)^{u_1},(-1)^{u_2},\ldots, (-1)^{u_n})=\prod_{i\in A}(1+(-1)^{u_i})\\
=\prod_{i\in A}(2-2u_i)=2^{|A|}\prod_{i\in A}(1-u_i)=2^{|A|}\chi(v|A).
\end{multline}
The result follows then from Eq. (5.2).

$(2)$ Let $B\neq \emptyset$. By Lemma \ref{lem2}, the length of the code $\mathcal{C}_{D}$ is $2^n-|\mathcal{I}(\mathbb{P})|=2^n-2^m-2^{|B|}+1$ and   
\begin{multline}
\mathcal{H}_{\mathcal{I}(\mathbb{P})}((-1)^{u_1},\ldots, (-1)^{u_n})
=\prod_{i\in 1}^m(1+(-1)^{u_i})+(-1)^{u_1+\cdots+u_m}(\prod_{j\in B}(1+(-1)^{u_j})-1)\\
=2^m\chi(v|[m])+(-1)^{wt(v)}(2^{|B|}\chi(w|B)-1).
\end{multline}
We proceed with the proof by considering the following three cases.

(i) If $v=0$, then $\chi(v|[m])=1$ and $\mathcal{H}_{\mathcal{I}(\mathbb{P})}((-1)^{u_1},\ldots, (-1)^{u_n})=2^{m}+2^{|B|}\chi(w|B)-1$.

(ii) If $v\neq 0$ and $wt(v)$ is odd, then $\chi(v|[m])=0$ and $\mathcal{H}_{\mathcal{I}(\mathbb{P})}((-1)^{u_1},\ldots, (-1)^{u_n})=-2^{|B|}\chi(w|B)+1$.
 
(iii)  If $v\neq 0$ and $wt(v)$ is even, then $\chi(v|[m])=0$ and $\mathcal{H}_{\mathcal{I}(\mathbb{P})}((-1)^{u_1},\ldots, (-1)^{u_n})=2^{|B|}\chi(w|B)-1$.

The result follows then from Eq. (5.2).
\end{proof}

\begin{rem} {\rm 
Let us discuss the parameters of the code $\mathcal{C}_D$ in Theorem 5.1. 

$(1)$ The parameters of the code  $\mathcal{C}_D$ in Theorem 5.1 (1) are $[2^n-2^{|A|},n,2^{n-1}-2^{|A|-1}]$ and these are the same as that in \cite{HLL} constructed from $\mathbb{H}(m,m)$. 

$(2)$ The parameters of the code $\mathcal{C}_{D}$ in Theorem 5.1 (2) are $[2^n-2^m-2^{|B|}+1,k]$, 
where $k=n-1$ or $n$. For instance, if $(m,|B|)=(1,n-1)$, then $2^{n-1}+1-2^{m-1}-2^{|B|}=0$, and  its dimension is $n-1$ in this case.
}
\end{rem}

Let $I_1=A_1\cup B_1$ and $I_2=A_2\cup B_2$ be two distinct order ideals of $\mathbb{H}(m,n)$, where $A_i\subseteq [m]$, $B_i\subseteq [n]\setminus[m]$, $i=1,2$. Here if $B_1=\emptyset$ and $B_2\neq \emptyset$, then  $I_1\subset I_2$ and   $D=(I_2(\mathbb{P}))^c=(\mathcal{I}(\mathbb{P}))^c$,  where $\mathcal{I}=\{I_1,I_2\}$. Note that in this case the code in Eq. (5.1) has been explored in Theorem 5.1 (2). Hence it  suffices to consider the following theorem for the case of two distinct order ideals.

\begin{thm} Let $\mathbb{H}(m,n)$ be a hierarchical poset with two levels. Let $I_1=A_1\cup B_1$ and $I_2=A_2\cup B_2$ be two distinct order ideals of $\mathbb{H}(m,n)$, where $A_i\subseteq [m]$, $B_i\subseteq [n]\setminus[m]$, $i=1,2$ and $I_1\not\subseteq I_2$, $I_2\not\subseteq I_1$. Set $\mathcal{I}=\{I_1,I_2\}$. 

$(1)$ If $B_1=B_2=\emptyset$, then the length of $\mathcal{C}_{D}$ is $2^n-2^{|A_1|}-2^{|A_2|}-2^{|A_1\cap A_2|}$ and its weight distribution is given in Table 3. 
\begin{table}  [h]
\caption{Theorem 5.3 (1)}  
\begin{center}  
\begin{tabu} to 0.95\textwidth{X[2.7,c]|X[3,c]}  
\hline 
\rm{Weight}&\rm{Frequency}\\ 
\hline
$0$&$1$\\ 
\hline
$2^{n-1}$&$2^{n-|A_1\cup A_2|}-1$\\ 
\hline
$2^{n-1}-2^{|A_2|-1}$& $ 2^{n-|A_1\cup A_2|}(2^{|A_1\backslash  A_2|}-1)$ \\ 
\hline
$2^{n-1}-2^{|A_1|-1}$& $2^{n-|A_1\cup A_2|}(2^{|A_2\backslash  A_1|}-1)$ \\ 
\hline
$2^{n-1}-2^{|A_1|-1}-2^{|A_2|-1}$& $2^{n-|A_1\cup A_2|}(2^{|A_1\backslash  A_2|}-1)(2^{|A_2\backslash  A_1|}-1)$ \\ 
\hline
$2^{n-1}-2^{|A_1|-1}-2^{|A_2|-1}+2^{|A_1\cap A_2|-1}$&$2^{n-|A_1\cup A_2|}(2^{|A_1\cap A_2|}-1)2^{|A_1\backslash  A_2|+|A_2\backslash  A_1|}$\\
\hline
\end{tabu}  
\end{center}  
\end{table} 

$(2)$ If $B_1\neq \emptyset$ and $B_2\neq \emptyset$, then the length of $\mathcal{C}_{D}$ is $2^n-2^{|B_1|}-2^{|B_2|}-2^{|B_1\cap B_2|}$ and its weight distribution is given in Table 4.
\begin{table}  [h]
\caption{Theorem 5.3 (2) }  
\begin{center}  
\begin{tabu} to 1\textwidth{X[4,c]|X[3,c]}  
\hline 
\rm{Weight}&\rm{Frequency}\\ 
\hline
$0$&$1$\\ 
\hline
$2^{n-1}$&$2^{n-m-|B_1\cup B_2|}-1$\\ 
\hline
$2^{n-1}-2^{|B_1|-1}$& $2^{n-m-|B_1\cup B_2|}(2^{|B_1\backslash B_2|}-1)$ \\ 
\hline
$2^{n-1}-2^{|B_2|-1}$& $2^{n-m-|B_1\cup B_2|}(2^{|B_2\backslash B_1|}-1)$ \\ 
\hline
$2^{n-1}-2^{|B_1|-1}-2^{|B_2|-1}$& $2^{n-m-|B_1\cup B_2|}(2^{|B_1\backslash B_2|}-1)(2^{|B_2\backslash B_1|}-1)$ \\ 
\hline
$2^{n-1}-2^{|B_1|-1}-2^{|B_2|-1}+2^{|B_1\cap B_2|-1}$&$ 2^{n-m-|B_1\cup B_2|}(2^{|B_1\cap B_2|}-1)2^{|B_1\backslash B_2 |+|B_2\backslash B_1 |}$\\
\hline
$2^{n-1}-2^{m-1}+1-2^{|B_1|}-2^{|B_2|}+2^{|B_1\cap B_2|}$&$2^{n-1-|B_1\cup B_2|}$\\ 
\hline
$2^{n-1}-2^{m-1}+1-2^{|B_2|}-2^{|B_1|-1}+2^{|B_1\cap B_2|}$&$2^{n-1-|B_1\cup B_2|}(2^{|B_1\backslash B_2|}-1)$\\ 
\hline
$2^{n-1}-2^{m-1}+1-2^{|B_1|}-2^{|B_2|-1}+2^{|B_1\cap B_2|}$&$2^{n-1-|B_1\cup B_2|}(2^{|B_2\backslash B_1|}-1)$\\ 
\hline
$2^{n-1}-2^{m-1}+1-2^{|B_1|-1}-2^{|B_2|-1}+2^{|B_1\cap B_2|}$&$2^{n-1-|B_1\cup B_2|}(2^{|B_1\backslash B_2|}-1)(2^{|B_2\backslash B_1|}-1)$\\
\hline
$2^{n-1}-2^{m-1}+1-2^{|B_1|-1}-2^{|B_2|-1}+2^{|B_1\cap B_2|-1}$&$2^{n-1-|B_1\cup B_2|}(2^{|B_1\cap B_2|}-1)2^{|B_1\backslash B_2 |+|B_2\backslash B_1 |}$\\ 
\hline
$2^{n-1}-2^{m-1}$&$2^{n-1-|B_1\cup B_2|}-2^{n-m-|B_1\cup B_2|}$\\
\hline
$2^{n-1}-2^{m-1}-2^{|B_1|-1}$&$(2^{n-1-|B_1\cup B_2|}-2^{n-m-|B_1\cup B_2|})(2^{|B_1\backslash B_2|}-1)$\\
\hline
$2^{n-1}-2^{m-1}-2^{|B_2|-1}$&$(2^{n-1-|B_1\cup B_2|}-2^{n-m-|B_1\cup B_2|})(2^{|B_2\backslash B_2|}-1)$\\
\hline
$2^{n-1}-2^{m-1}-2^{|B_1|-1}-2^{|B_2|-1}$&$2^{n-m-|B_1\cup B_2|}(2^{m-1}-1)(2^{|B_1\backslash B_2|}-1)(2^{|B_2\backslash B_1|}-1)$\\
\hline
$2^{n-1}-2^{m-1}-2^{|B_1|-1}-2^{|B_2|-1}+2^{|B_1\cap B_2|-1}$&$2^{n-m-|B_1\cup B_2|}(2^{m-1}-1)(2^{|B_1\cap B_2|}-1)2^{|B_1\backslash B_2|+|B_2\backslash B_1|}$\\
\hline
\end{tabu}  
\end{center}  
\end{table} 
\end{thm}

\begin{proof} 
It suffices to determine the form of the order ideal $I_1\cap I_2$ by Theorem \ref{th1}. 

$(1)$ Let $B_1=B_2=\emptyset$. Then $I_1\cap I_2\subseteq [m]$. By Lemma \ref{lem2}, the length of $\mathcal{C}_{D}$ is $2^n-|\mathcal{I}(\mathbb{P})|=2^n-2^{|A_1|}-2^{|A_2|}-2^{|A_1\cap A_2|}$.
By Theorem \ref{th1}, we have
\begin{multline*}
\mathcal{H}_{\mathcal{I}(\mathbb{P})}((-1)^{u_1},\ldots, (-1)^{u_n})
=2^{|A_1|}\chi(v|A_1)+2^{|A_2|}\chi(v|A_2)-2^{|A_1\cap A_2|}\chi(v|A_1\cap A_2).
\end{multline*}
For two subsets $A_1, A_2$ of $2^{[n]}$, we set
$$\mathcal{U}_1=\{u\in \mathbb{F}_2^n: u\cap (A_1\cup A_2)=\emptyset\},$$ $$\mathcal{U}_2=\{u\in \mathbb{F}_2^n: u\cap A_1=\emptyset, u\cap(A_2\backslash A_1)\neq\emptyset\},$$
$$\mathcal{U}_3=\{u\in \mathbb{F}_2^n: u\cap A_2=\emptyset, u\cap(A_1\backslash A_2)\neq\emptyset\},$$
$$\mathcal{U}_4=\{u\in \mathbb{F}_2^n: u\cap(A_1\backslash A_2)\neq\emptyset, u\cap(A_1\cap A_2)=\emptyset, u\cap (A_2\backslash A_1)\neq \emptyset\},$$
$$\mathcal{U}_5=\{u\in \mathbb{F}_2^n: u\cap(A_1\cap A_2)\neq\emptyset \}.$$

Then 
\begin{align*}
\mathcal{H}_{\mathcal{I}(\mathbb{P})}((-1)^{u_1},\ldots, (-1)^{u_n})=\left\{\begin{array}{lll}
2^{|A_1|}+2^{|A_2|}-2^{|A_1\cap A_2|} &\mbox{if $u\in \mathcal{U}_1$,}\\
2^{|A_1|}-2^{|A_1\cap A_2|}&\mbox{if $u\in \mathcal{U}_2$,} \\
2^{|A_2|}-2^{|A_1\cap A_2|}&\mbox{if $u\in \mathcal{U}_3$,} \\
-2^{|A_1\cap A_2|}&\mbox{if $u\in \mathcal{U}_4$,} \\
0&\mbox{if $u\in \mathcal{U}_5$.}\\
\end{array}\right.
\end{align*}
The result follows then from Lemma \ref{lem2} and Eq. (5.2).

$(2)$ Let $B_1\neq \emptyset$ and $B_2\neq \emptyset$. Then $I_1\cap I_2=[m]\cup (B_1\cap B_2)$. By Lemma \ref{lem2}, the length of $\mathcal{C}_{D}$ is $2^n-|\mathcal{I}(\mathbb{P})|=2^n-2^{|B_1|}-2^{|B_2|}-2^{|B_1\cap B_2|}$.
By Theorem \ref{th1}, we have that
$\mathcal{H}_{\mathcal{I}(\mathbb{P})}((-1)^{u_1},\ldots, (-1)^{u_n})$ is equal to
\begin{eqnarray*}
&&2^{m}\chi(v|[m])+(-1)^{wt(v)}2^{|B_1|}\chi(w|B_1)-(-1)^{wt(v)}\\
&+&2^{m}\chi(v|[m])+(-1)^{wt(v)}2^{|B_2|}\chi(w|B_2)-(-1)^{wt(v)}\\
&-&2^{m}\chi(v|[m])-(-1)^{wt(v)}2^{|B_1\cap B_2|}\chi(v|B_1\cap B_2)+(-1)^{wt(v)}\\
&=&2^{m}\chi(v|[m])+(-1)^{wt(v)}2^{|B_1|}\chi(w|B_1)+(-1)^{wt(v)}2^{|B_2|}\chi(w|B_2)\\
&-&(-1)^{wt(v)}2^{|B_1\cap B_2|}\chi(v|B_1\cap B_2)-(-1)^{wt(v)}.
\end{eqnarray*}
We proceed with the proof by considering the following three cases.

(i) If $v=0$, then $\chi(v|[m])=1$ and 
\begin{multline*}
\mathcal{H}_{\mathcal{I}(\mathbb{P})}((-1)^{u_1},\ldots, (-1)^{u_n})\\
=2^{m}-1+2^{|B_1|}\chi(w|B_1)+2^{|B_2|}\chi(w|B_2)-2^{|B_1\cap B_2|}\chi(v|B_1\cap B_2).   
\end{multline*}

(ii) If $v\neq0$ and $wt(v)$ is odd, then $\chi(v|[m])=0$ and 
\begin{multline*}
\mathcal{H}_{\mathcal{I}(\mathbb{P})}((-1)^{u_1},\ldots, (-1)^{u_n})\\
=1-2^{|B_1|}\chi(w|B_1)-2^{|B_2|}\chi(w|B_2)+2^{|B_1\cap B_2|}\chi(v|B_1\cap B_2).
\end{multline*}

(iii) If $v\neq0$ and $wt(v)$ is even, then $\chi(v|[m])=0$ and 
\begin{multline*}
\mathcal{H}_{\mathcal{I}(\mathbb{P})}((-1)^{u_1},\ldots, (-1)^{u_n})\\
=-1+2^{|B_1|}\chi(w|B_1)+2^{|B_2|}\chi(w|B_2)-2^{|B_1\cap B_2|}\chi(v|B_1\cap B_2).
\end{multline*}

The result follows then from the case by case consideration as in $(1)$, Lemma \ref{lem2} and Eq. (5.2).   
\end{proof}

\begin{rem}
{\rm

Let us discuss the parameters of the code $\mathcal{C}_D$ in Theorem 5.3. 

$(1)$ The parameters of the code $\mathcal{C}_D$ in Theorem 5.3 (1) are  $[2^n-2^{|A_1|}-2^{|A_2|}-2^{|A_1\cap A_2|},n,2^{n-1}-2^{|A_1|-1}-2^{|A_2|-1}]$ and these are the same as that in \cite{HLL} constructed from $\mathbb{H}(m,m)$. 

$(2)$ The parameters of  the code $\mathcal{C}_{D}$ are $[2^n-2^{|B_1|}-2^{|B_2|}-2^{|B_1\cap B_2|},n,2^{n-1}-2^{m-1}+1-2^{|B_1|}-2^{|B_2|}+2^{|B_1\cap B_2|}]$.
}
\end{rem}
\subsection{Linear codes from Boolean functions}$~$

By a Boolean function we mean a function from $\Bbb F_2^n$ 
to $\Bbb F_2$. Let $f$ be a  Boolean function from $\Bbb F_2^n$ to $\Bbb F_2$ such that  $f(0)=0$ but $f(u)=1$ for at least one $u\in \Bbb F_2^n$. We introduce a linear code associated with $f$ as follows:

\begin{equation}
\mathcal{C}_f=\{c_f(s,u)=(sf(x)+u\cdot x)_{x\in \mathbb{F}_2^{n*}}:s\in \Bbb F_2,u\in \Bbb F_2^n\}.
\end{equation}
Then the code $\mathcal{C}_f$ has the length $2^n-1$ and its dimension is at most $n+1$.
The construction of linear codes from Boolean functions can be found in \cite{CH, D2, DHZ, M2, WHWK}.

Let $\mathbb{P}$ be a poset on $[n]$ and $\mathcal{I}=\{I_1,\ldots,I_k\}\subseteq\mathcal{O}_{\mathbb{P}}$. 
Let $f$ be a Boolean function from $\Bbb F_2^n$ to $\Bbb F_2$ with support $\mathcal{I}(\mathbb{P})\backslash \{\emptyset\}$, that is, $f(u)=1$ for all $u\in\mathbb{F}^n_2$ such that $\operatorname{supp}(u)\in\mathcal{I}(\mathbb{P})\backslash \{\emptyset\}$.
By \cite[Lemma 3]{CH}, the Walsh-Hadamard transform defined by $S_f(u)=\sum_{v\in\mathbb{F}_2^n}(-1)^{f(v)+u\cdot v}$ becomes
\[
S_{f}(u)=2^n\delta_{0,u}+2-2\mathcal{H}_{\mathcal{I}(\mathbb{P})}((-1)^{u_1},(-1)^{u_2},\ldots, (-1)^{u_n}).
\]
Then 
\begin{multline*} 
wt(c_f(s,u))=2^n-1-\frac12\sum_{y\in \Bbb F_2}\sum_{x\in \Bbb F_2^{n*}}(-1)^{y(sf(x)+ux)}\\
=2^{n-1}-\frac12S_{sf}(u)
=\left\{\begin{array}{ll}
2^{n-1}-2^{n-1}\delta_{0,u} &\mbox{if $s=0$,}\\
2^{n-1}(1-\delta_{0,u})-1+\mathcal{H}_{\mathcal{I}(\mathbb{P})}((-1)^{u_1},\ldots, (-1)^{u_n}) & \mbox{if $s=1$.}\\
\end{array}\right.
\end{multline*}
It follows that
\begin{equation}
wt(c_f(s,u))=2^{n-1}(1-\delta_{0,u})+\delta_{1,s}\left(\mathcal{H}_{\mathcal{I}(\mathbb{P})}((-1)^{u_1},\ldots,(-1)^{u_n}) -1\right).
\end{equation}

In this subsection, we also assume that $\mathbb{P}$ is a hierarchical poset $\mathbb{H}(m,n)$ with two levels. We determine the weight distribution of the code  $\mathcal{C}_f$ defined in (5.5) only when $k=1$ and $k=2$.

\begin{thm}
Let $\mathbb{H}(m,n)$ be a hierarchical poset with two levels and $I=A\cup B$ an order ideal of $\mathbb{H}(m,n)$ for $A\subseteq [m]$, $B\subseteq [n]\setminus[m]$. Set $\mathcal{I}=\{I\}$. 


$(1)$ If $B=\emptyset$, then the weight distribution of  the code $\mathcal{C}_f$ is given in Table 5.
\begin{table}  [h]
\caption{Theorem 5.5 (1)}   
\begin{tabu} to 0.4\textwidth{X[2,c]|X[2,c]}  
\hline 
\rm{Weight}&\rm{Frequency}\\ 
\hline
$0$&$1$\\ 
\hline
$2^{n-1}$&$2^{n}-1$\\
\hline
$2^{|A|}-1$&$1$\\
\hline
$2^{n-1}-1+2^{|A|}$& $2^{n-|A|}-1$ \\ 
\hline
$2^{n-1}-1$& $2^n-2^{n-|A|}$ \\ 
\hline
\end{tabu}   
\end{table} 

$(2)$ If $B\neq \emptyset$, then the weight distribution of the code $\mathcal{C}_f$ is given in Table 6.
\begin{table}  [h]
\caption{Theorem 5.5 (2)}  
\begin{center}  
\begin{tabu} to 0.8\textwidth{X[2,c]|X[2,c]}  
\hline 
\rm{Weight}&\rm{Frequency}\\
\hline
$0$&$1$\\ 
\hline
$2^{n-1}$&$2^{n}-1+2^{n-1}-2^{n-1-|B|}$\\
\hline
$2^m+2^{|B|}-2$&$1$\\
\hline
$2^{n-1}-2+2^{m}+2^{|B|}$&$2^{n-m-|B|}-1$\\
\hline
$2^{n-1}-2+2^{m}$& $2^{n-m}-2^{n-m-|B|}$ \\
\hline
$2^{n-1}-2^{|B|}$&$2^{n-1-|B|}$ \\
\hline
$2^{n-1}-2+2^{|B|}$&$2^{n-1-|B|}-2^{n-m-|B|}$ \\
\hline
$2^{n-1}-2$&$2^{n-1}-2^{n-1-|B|}-2^{n-m}+2^{n-m-|B|}$ \\
\hline
\end{tabu}  
\end{center}  
\end{table} 
\end{thm}

\begin{proof} 
To obtain the weight  of the codeword $c_f(s,u)$, it suffices to compute the value of $\mathcal{H}_{\mathcal{I}(\mathbb{P})}((-1)^{u_1},\ldots, (-1)^{u_n})$ by Eq. (5.6), where $u=(u_1, \ldots, u_n)$. The results follow then from Eqs. (5.3) and (5.4).
\end{proof}

\begin{rem}  {\rm 
Let us discuss the parameters of the code $\mathcal{C}_f$ in Theorem 5.5. 

$(1)$ The parameters of the code  $\mathcal{C}_f$ in Theorem 5.5 (1) are $[2^n-1,k]$, where $k=n$ or $n+1$. For instance, if $|A|=1$, then $2^{|A|}-1=0$, and its dimension is $n$ in this case. These codes are the same as that in \cite{HLL} constructed from $\mathbb{H}(m,m)$.

$(2)$ The parameters of the code $\mathcal{C}_{f}$ in Theorem 5.5 (2) are $[2^n-1,k]$, 
where $k=n$ or $n+1$. For instance, if $(m,|B|)=(1,n-1)$, then $2^{n-1}-2^{|B|}=0$, and  its dimension is $n$ in this case.
}
\end{rem}

\begin{thm}  Let $\mathbb{H}(m,n)$ be a hierarchical poset with two levels. Let $I_1=A_1\cup B_2$ and $I_2=A_2\cup B_2$ be two distinct order ideals of $\mathbb{H}(m,n)$, where $A_i\subseteq [m]$, $B_i\subseteq [n]\setminus[m]$, $i=1,2$ and $I_1\not\subseteq I_2$, $I_2\not\subseteq I_1$. Set $\mathcal{I}=\{I_1,I_2\}$. 

$(1)$ If $B_1=B_2=\emptyset$, then the weight distribution of the code $\mathcal{C}_f$ is given  in Table 7.

$(2)$ If $B_1\neq \emptyset$ and $B_2\neq \emptyset$, then the weight distribution of  the code $\mathcal{C}_f$ is given in Table 8.
\end{thm} 
\begin{proof} 
By Eq. (5.6), to obtain the weight  of the codeword $c_f(s,u)$, it suffices to compute the value of $\mathcal{H}_{\mathcal{I}(\mathbb{P})}((-1)^{u_1},\ldots, (-1)^{u_n})$, where $u=(u_1, \ldots, u_n)$. 
The results follow then from the proof of Theorem 5.3.
\end{proof}

\begin{rem}  {\rm 
Let us discuss the parameters of the code $\mathcal{C}_f$ in Theorem 5.7. 

$(1)$ The parameters of the code  $\mathcal{C}_f$ in Theorem 5.7 (1) are $[2^n-1,n+1,2^{|A_1|}+2^{|A_2|}-2^{|A_1\cap A_2|}-1]$ and these codes are the same as that in \cite{HLL} constructed from $\mathbb{H}(m,m)$.

$(2)$ The parameters of the code $\mathcal{C}_{f}$ in Theorem 5.7 (2) are $[2^n-1,n+1,2^m-2+2^{|B_1|}+2^{|B_2|}-2^{|B_1\cap B_2|}]$.
}
\end{rem}

\begin{table}  [h]
\caption{Theorem 5.7 (1)}  
\begin{center}  
\begin{tabu} to 0.8\textwidth{X[2,c]|X[3,c]}  
\hline 
\rm{Weight}&\rm{Frequency}\\ 
\hline
$0$&$1$\\ 
\hline
$2^{n-1}$&$2^{n}-1$\\
\hline
$2^{|A_1|}+2^{|A_2|}-2^{|A_1\cap A_2|}-1$&$1$\\ 
\hline
$2^{n-1}+2^{|A_1|}+2^{|A_2|}-2^{|A_1\cap A_2|}-1$&$2^{n-|A_1\cup A_2|}-1$\\ 
\hline
$2^{n-1}+2^{|A_1|}-2^{|A_1\cap A_2|}-1$& $ 2^{n-|A_1\cup A_2|}(2^{|A_1\backslash  A_2|}-1)$ \\ 
\hline
$2^{n-1}+2^{|A_2|}-2^{|A_1\cap A_2|}-1$& $2^{n-|A_1\cup A_2|}(2^{|A_2\backslash  A_1|}-1)$ \\ 
\hline
$2^{n-1}-2^{|A_1\cap A_2|}-1$& $2^{n-|A_1\cup A_2|}(2^{|A_1\backslash  A_2|}-1)(2^{|A_2\backslash  A_1|}-1)$ \\ 
\hline
$2^{n-1}-1$&$2^{n-|A_1\cup A_2|}(2^{|A_1\cap A_2|}-1)2^{|A_1\backslash  A_2|+|A_2\backslash  A_1|}$\\
\hline
\end{tabu}  
\end{center}  
\end{table}

\begin{table}   [h]
\caption{Theorem 5.7 (2) }  
\begin{tabu} to 1\textwidth{X[2,c]|X[3.2,c]}  
\hline 
\rm{Weight}&\rm{Frequency}\\ 
\hline
$0$&$1$\\ 
\hline
$2^m-2+2^{|B_1|}+2^{|B_2|}-2^{|B_1\cap B_2|}$&$1$\\ 
\hline
$2^{n-1}+2^m-2+2^{|B_1|}+2^{|B_2|}-2^{|B_1\cap B_2|}$&$2^{n-m-|B_1\cup B_2|}-1$\\ 
\hline
$2^{n-1}+2^m-2+2^{|B_2|}-2^{|B_1\cap B_2|}$& $2^{n-m-|B_1\cup B_2|}(2^{|B_1\backslash B_2|}-1)$ \\ 
\hline
$2^{n-1}+2^m-2+2^{|B_1|}-2^{|B_1\cap B_2|}$& $2^{n-m-|B_1\cup B_2|}(2^{|B_2\backslash B_1|}-1)$ \\ 
\hline
$2^{n-1}+2^m-2-2^{|B_1\cap B_2|}$& $2^{n-m-|B_1\cup B_2|}(2^{|B_1\backslash B_2|}-1)(2^{|B_2\backslash B_1|}-1)$ \\ 
\hline
$2^{n-1}+2^m-2$&$ 2^{n-m-|B_1\cup B_2|}(2^{|B_1\cap B_2|}-1)2^{|B_1\backslash B_2 |+|B_2\backslash B_1 |}$\\
\hline
$2^{n-1}-2^{|B_1|}-2^{|B_2|}+2^{|B_1\cap B_2|}$&$2^{n-1-|B_1\cup B_2|}$\\ 
\hline
$2^{n-1}-2^{|B_2|}+2^{|B_1\cap B_2|}$&$2^{n-1-|B_1\cup B_2|}(2^{|B_2\backslash B_1|}-1)$\\ 
\hline
$2^{n-1}-2^{|B_1|}+2^{|B_1\cap B_2|}$&$2^{n-1-|B_1\cup B_2|}(2^{|B_1\backslash B_2|}-1)$\\ 
\hline
$2^{n-1}+2^{|B_1\cap B_2|}-2^{|B_1|-1}-2^{|B_2|-1}$&$2^{n-1-|B_1\cup B_2|}(2^{|B_1\backslash B_2|}-1)(2^{|B_2\backslash B_1|}-1)$\\
\hline
$2^{n-1}$&$2^{n}-1+2^{n-1-|B_1\cup B_2|}(2^{|B_1\cap B_2|}-1)2^{|B_1\backslash B_2 |+|B_2\backslash B_1 |}$\\ 
\hline
$2^{n-1}-2+ 2^{|B_1|}+2^{|B_2|}-2^{|B_1\cap B_2|}$&$2^{n-1-|B_1\cup B_2|}-2^{n-m-|B_1\cup B_2|}$\\
\hline
$2^{n-1}-2+2^{|B_2|}-2^{|B_1\cap B_2|}$&$(2^{n-1-|B_1\cup B_2|}-2^{n-m-|B_1\cup B_2|})(2^{|B_1\backslash B_2|}-1)$\\
\hline
$2^{n-1}-2+ 2^{|B_1|}-2^{|B_1\cap B_2|}$&$(2^{n-1-|B_1\cup B_2|}-2^{n-m-|B_1\cup B_2|})(2^{|B_2\backslash B_2|}-1)$\\
\hline
$2^{n-1}-2-2^{|B_1\cap B_2|}$&$2^{n-m-|B_1\cup B_2|}(2^{m-1}-1)(2^{|B_1\backslash B_2|}-1)(2^{|B_2\backslash B_1|}-1)$\\
\hline
$2^{n-1}-2$&$2^{n-m-|B_1\cup B_2|}2^{|B_1\backslash B_2|+|B_2\backslash B_1|}(2^{m-1}-1)(2^{|B_1\cap B_2|}-1)$\\
\hline
\end{tabu}  
\end{table}

\section{optimal and minimal linear codes}

In this section, we find some infinite families of optimal and minimal binary linear codes based on the results in Section 5.



\begin{thm}
Let $\mathcal{C}_{D}$ be a linear code constructed from Theorem 5.1 (2). 

$(1)$ If $|B|=1<m-1\le n-2$, then the code $\mathcal{C}_{D}$ is a Griesmer code with parameters $[2^n-1-2^m, n, 2^{n-1}-1-2^{m-1}]$. 

$(2)$ If $m=1$ and $|B|=n-1$, then the code $\mathcal{C}_{D}$  is a Griesmer code with parameters $[2^{n-1}-1, n-1, 2^{n-2}]$. 

$(3)$ The code $\mathcal{C}_{D}$ is minimal if and only if $(m,|B|)\notin \{(1,n-1),(1,n-2), (2,n-2), (n-1,1)\}$.




\end{thm}

\begin{proof} 


$(1)$ By Theorem 5.1 $(2)$, the minimum distance of the code  $\mathcal{C}_{D}$ is $2^{n-1}-1-2^{m-1}$, and so \begin{eqnarray*}
&&\sum_{i=0}^{n-1}\bigg\lceil {\frac{2^{n-1}-1-2^{m-1}}{2^i}} \bigg\rceil \\
&=&(2^{n-1}-1-2^{m-1})+(2^{n-2}-2^{m-2})+\cdots+(2^{n-m}-1)+2^{n-m-1}+\cdots+1\\
&=&(1+2+\cdots+2^{n-1})+(1+2+\cdots+2^{m-1})-1
\\&=&(2^n-1)-(2^m-1)-1= 2^n-1-2^m,
\end{eqnarray*} 
which implies that the code $\mathcal{C}_{D}$ is a Griesmer code. 

$(2)$ By Theorem 5.1 $(2)$,  the code $\mathcal{C}_{D}$ has parameters $[2^{n-1}-1, n-1, 2^{n-2}]$. Then \begin{eqnarray*}
&&\sum_{i=0}^{n-2}\bigg\lceil {\frac{2^{n-2}}{2^i}} \bigg\rceil =2^{n-1}-1,
\end{eqnarray*} which implies that the code $\mathcal{C}_{D}$ is a Griesmer code.


$(3)$ By Table 2, the code $\mathcal{C}_{D}$ is  a  at most six-weight linear code with $w_1=2^{n-1}, w_2=2^{n-1}-2^{|B|-1},w_3=2^{n-1}+1-2^{m-1}-2^{|B|},w_4=2^{n-1}+1-2^{m-1}-2^{|B|-1},w_5=2^{n-1}-2^{m-1}$ and $w_6=2^{n-1}-2^{m-1}-2^{|B|-1}$. The minimum weight $d$ of the code $\mathcal{C}_{D}$ is either $w_3$ or $w_6$ and its maximum weight is $2^{n-1}$. 
Suppose that $d=w_3>0~(\Leftrightarrow (m,|B|)\neq(1,n-1))$.
The sufficient condition of Lemma 2.2 to be minimal is that $2^{n-1}>2^m+2^{|B|+1}-2$. It follows that if $(m,|B|)\notin \{(1,n-2), (2,n-2), (n-1,1)\}$, then the code $\mathcal{C}_{D}$ is minimal. 
On the other hand, by Lemma 2.3, the code $\mathcal{C}_{D}$ is minimal if and only if $w_i+w_j\neq w_k$ for any $i,j,k\in \{1,2,\ldots,6\}$. The tedious computations show that if $(m,|B|)\in \{(1,n-1),(1,n-2), (2,n-2), (n-1,1)\}$, then the latter condition does not hold. In fact, if $(m,|B|)=(1,n-1),(1,n-2), (2,n-2),(n-1,1)$, then $2w_4=w_1,2w_3=w_1,2w_3=w_5,2w_5=w_1$, respectively. 

This completes the proof. \end{proof}



\begin{thm} 

Let $\mathcal{C}_{D}$ be a linear code constructed from Theorem 5.3 (2). If $|B_1|=|B_2|=1$ and $1< m\le n-2$, then the code $\mathcal{C}_{D}$ is a distance-optimal linear code with parameters $[2^n-2^m-2,n,2^{n-1}-2^{m-1}-2]$.
In this case, $\mathcal{C}_{D}$ is minimal if $n\ge 4$.
\end{thm}

\begin{proof}

Note that in Theorem 5.3 (2), $I_1\not\subseteq I_2$ and $I_2\not\subseteq I_1$. Then $B_1\cap B_2=\emptyset$ as $|B_1|=|B_2|=1$. In this case, the code $\mathcal{C}_{D}$ is a at most six-weight linear code with $w_1=2^{n-1}$, $w_2=2^{n-1}-1$, $w_3=2^{n-1}-2$, $w_4=2^{n-1}-2^{m-1}$, $w_5=2^{n-1}-2^{m-1}-1$ and $w_6=2^{n-1}-2^{m-1}-2$. It follows that the code $\mathcal{C}_{D}$ has parameters $[2^n-2^m-2,n,2^{n-1}-2^{m-1}-2]$. By using the Griesmer bound, we can show that it is distance-optimal. Due to $m\le n-2$ and $n\ge 4$, we get $2w_6=2^n-2^m-4\ge 2^n-2^{n-2}-4=2^{n-1}+2^{n-2}-4>2^{n-1}=w_1$. Then the code $\mathcal{C}_{D}$ is minimal by Lemma 2.2. 
This completes the proof.
\end{proof}

\begin{thm} 
Let $\mathcal{C}_f$ be a linear code constructed from Theorem 5.5 (2).

$(1)$ If $m=n-1$ and $|B|=1$, then the code $\mathcal{C}_{f}$  is an almost optimal linear code with parameters $[2^n-1, n+1, 2^{n-1}-2]$. 

$(2)$ If $m+|B|=n\geq5$ and $\max\{m,|B|\}\le n-2$, then the code $\mathcal{C}_{f}$ is minimal violating the condition of Aschikhmin-Barg.
\end{thm}

\begin{proof} 
$(1)$ The parameters of the code $\mathcal{C}_f$ follows from Table 6.  
We see from
\begin{align*}&&\sum_{i=0}^{n}\bigg\lceil {\frac{2^{n-1}-1}{2^i}} \bigg\rceil =1+ \sum_{i=0}^{n-1}\bigg\lceil {\frac{2^{n-1}-1}{2^i}}\bigg\rceil =1+(2^n-1)-1=2^n-1
\end{align*}
that the code $\mathcal{C}_f$ is almost optimal.

$(2)$ The weights of the code $\mathcal{C}_{f}$ are as follows:
$$w_1=2^{n-1}, w_2=2^m+2^{|B|}-2, w_3=2^{n-1}-2^{|B|},w_4=2^{n-1}-2,$$
$$w_5=2^{n-1}-2+2^m, w_6=2^{n-1}-2+2^{|B|}$$
because the frequency corresponding to the weight $2^{n-1}-2+2^m+2^{|B|}$ equals the zero by our assumption that $m+|B|=n$. 
We see from Eq. (5.6) that for any $u\in \Bbb F_2^n$, 
\begin{equation}
wt(c_f(0,u))\in \{0,w_1\}\mbox{ or } wt(c_f(1,u))\in\{w_2,w_3, w_4, w_5,w_6\}.
\end{equation}
By Lemma 2.3, the code $\mathcal{C}_{f}$ is minimal if and only if for each pair of distinct nonzero codewords  $c_f(s,u)$ and  $ c_f(s',u')$ in the code $\mathcal{C}_{f}$, we should have 
\begin{equation}
wt(c_f(s+s',u+u'))=wt(c_f(s,u)+c_f(s',u'))\neq wt(c_f(s,u))-wt(c_f(s',u')).
\end{equation}
It follows from Eqs. (6.1) and (6.2) that to show the code $\mathcal{C}_{f}$ is minimal, it suffices to verify that $w_1+ w_i\neq w_j$ for any $i,j\in \{2,3,4,5,6\}$.

By the conditions that $m+|B|=n$ and $\max\{m,|B|\}\le n-2$, we have $|B|\geq2$. Let $w_{\max}$ denote the maximum nonzero Hamming weights in the code $\mathcal{C}_f$. Then $w_{\max}=2^{n-1}-2+2^{\max\{m,|B|\}}$, so that 
 $$2^{n-1}+w_2=2^{n-1}+2^m+2^{|B|}-2>w_{\max},$$
 $$2^{n-1}+w_3=2^{n-1}+2^{n-1}+2^{|B|}>w_{\max},$$  $$2^{n-1}+w_4=2^{n}-2>2^{n}-2^{|B|}>w_{\max},$$   
 $2^{n-1}+w_5>2^{n-1}+w_4$ and $2^{n-1}+w_6>2^{n-1}+w_4$. This prove that the code $\mathcal{C}_{f}$ is minimal. 
It remains to show that ${w_{\min}}/{w_{\max}}<1/2$.
We proceed with the proof by considering the following three cases.

(i) If $|B|\leq m$, then $w_{\min}=2^m+2^{|B|}-2$ and $w_{\max}=2^{n-1}+2^m-2$ because $n\geq5$, and so $w_{\max}-2w_{\min}=2^{n-1}+2^m-2-(2^{m+1}+2^{|B|+1}-4)=2^{n-1}+2-2^m-2^{|B|+1}>0.$
 
(ii) If $m<|B|<n-2$, then $w_{\min}=2^m+2^{|B|}-2$ and $w_{\max}=2^{n-1}+2^{|B|}-2$ because $n\geq5$, and so $w_{\max}-2w_{\min}=2^{n-1}+2^{|B|}-2-(2^m+2^{|B|}-2)=2^{n-1}-2^{|B|}>0$.
 
(iii) If $|B|=n-2$ and $m=2$, then $w_{\min}=2^{n-1}-2^{|B|}=2^{n-2}$ and $w_{\max}=2^{n-1}+2^{|B|}-2$, and so $w_{\max}-2w_{\min}=2^{n-1}+2^{|B|}-2-2^{n-1}=2^{|B|}-2>0$.

This completes the proof.
\end{proof}

\begin{thm} 
Let $\mathcal{C}_f$ be a linear code constructed from Theorem 5.7 (2). If $B_1\cap B_2=\emptyset$, $|B_1|+|B_2|=n-m$, and $\max\{|B_1|,|B_2|\}\le n-2$, then the code $\mathcal{C}_{f}$ is minimal violating the condition of Aschikhmin-Barg with parameters $[2^n-1, n+1, 2^m-3+2^{|B_1|}+2^{|B_2|}]$.
\end{thm}
\begin{proof} The weights of the code $\mathcal{C}_{f}$ are as follows:
$$w_1=2^{n-1}, w_2=2^m+2^{|B_1|}+2^{|B_2|}-3, w_3=2^{n-1}+2^m-3-2^{|B_1|},
w_4=2^{n-1}+2^m-3-2^{|B_2|},
$$
$$w_5=2^{n-1}+2^m-3,w_6=2^{n-1}+1-2^{|B_1|},
w_7=2^{n-1}+1-2^{|B_2|},w_8=2^{n-1}+1,$$
$$w_9=2^{n-1}-3,w_{10}=2^{n-1}-3+2^{|B_1|}, w_{11}=2^{n-1}-3+2^{|B_2|},w_{12}=2^{n-1}-3+2^{|B_1|}+2^{|B_2|},$$
because the  numbers $2^{n-m-|B_1\cup B_2|}-1$, and $2^{|B_1\cap B_2|}-1$ equals the zero which appear in frequencies  in Table 8 by our assumptions that $B_1\cap B_2=\emptyset$ and $|B_1|+|B_2|=n-m$.

We see from Eq. (5.6) that for any $u\in \Bbb F_2^n$, 
\begin{equation}
wt(c_f(0,u))\in \{0,w_1\}\mbox{ or } wt(c_f(1,u))\in\{w_2,\cdots,w_{12}\}.
\end{equation}
By Lemma 2.3,  the code $\mathcal{C}_{f}$ is minimal if and only if for each pair of distinct nonzero codewords  $c_f(s,u)$ and  $ c_f(s',u')$ in  the code $\mathcal{C}_{f}$, we should have 
\begin{equation}
wt(c_f(s+s',u+u'))=wt(c_f(s,u)+c_f(s',u'))\neq wt(c_f(s,u))-wt(c_f(s',u')).
\end{equation}
It follows from Eqs. (6.3) and (6.4) that to show the code $\mathcal{C}_{f}$ is minimal, it suffices to verify that $w_1+ w_i\neq w_j$ for any $i,j\in \{2,\cdots, 12\}$.

Let $w_{\max}$ denote the maximum nonzero Hamming weights in the code $\mathcal{C}_f$.
By the conditions that $B_1\cap B_2=\emptyset$,  $|B_1|+|B_2|=n-m$, and $\max\{|B_1|,|B_2|\}\le n-2$, we have $w_{\max}=2^{n-1}+2^m-3-2^{\min\{|B_1|,|B_2|\}}$ or $2^{n-1}-3+2^{|B_1|}+2^{|B_2|}$,
so that
 $$2^{n-1}+w_2=2^{n-1}+2^m+2^{|B_1|}+2^{|B_2|}-3>w_{\max},$$
 $$2^{n-1}+w_6=2^{n-1}+2^{n-1}-2^{|B_1|}>w_{\max},$$
 $$2^{n-1}+w_7=2^{n-1}+2^{n-1}-2^{|B_2|}>w_{\max},$$
and for $i\in \{3,4,5, 8,\ldots, 12\}$, there exists $j\in\{2,6,7\}$ such that $2^{n-1}+w_i>2^{n-1}+w_j$. This prove that the code $\mathcal{C}_{f}$ is minimal. 
It remains to show that ${w_{\min}}/{w_{\max}}<1/2$. 

We proceed with the proof by considering the following two cases.
It is easy to check that  $w_{\min}=\min\{w_1, \ldots, w_{12}\}=w_3=2^m+2^{|B_1|}+2^{|B_2|}-3$.

(i) If $m\ge \max\{|B_1|, |B_2|\}$ or $\max\{|B_1|, |B_2|\}>m\ge\min\{|B_1|, |B_2|\}$, then  $w_{max}=2^{n-1}+2^m-3+2^{\max\{|B_1|, |B_2|\}}$. Then \begin{eqnarray*}&&w_{\max}-2w_{\min}
\\&=&2^{n-1}+2^m-3+2^{\max\{|B_1|, |B_2|\}}-(2^{m+1}+2^{|B_1|+1}+2^{|B_2|+1}-6)\\
&=&2^{n-1}+3-(2^m+2^{|B_1|+1}+2^{|B_2|+1}-2^{\max\{|B_1|, |B_2|\}})>0.\end{eqnarray*}
 
(ii) If $m<\min\{|B_1|, |B_2|\}$ ,  then $w_{\max}=2^{n-1}+2^{\max\{|B_1|, |B_2|\}}+2^{\min\{|B_1|, |B_2|\}}-3$. Then \begin{eqnarray*}&&w_{\max}-2w_{\min}
\\&=&2^{n-1}+2^{\min\{|B_1|, |B_2|\}}-3+2^{\max\{|B_1|, |B_2|\}}-(2^{m+1}+2^{|B_1|+1}+2^{|B_2|+1}-6)\\
&>&2^{n-1}+3-(2^m+2^{|B_1|+1}+2^{|B_2|+1}-2^{\max\{|B_1|, |B_2|\}})>0.\end{eqnarray*}

 This completes the proof.
\end{proof}

\begin{exa} Let $\mathbb{H}(2,5)$ be a hierarchical poset with two levels. Let $I=\{1,3,4\}$. Then the code defined in Eq. (5.1) has parameters $[25, 5,11]$ and its weight enumerator is given by $$1+4z^{11}+6z^{12}+12z^{13}+8z^{14}+z^{16}.$$ In fact, the optimal binary  linear code has parameter $[25,5,12]$, according to  \cite{G2}. 
\end{exa}

\begin{exa} Let $\mathbb{H}(2,5)$ be a hierarchical poset with two levels. Let $I=\{1,3,4,5\}$. Then the code $\mathcal{C}_{f}$ defined in Eq. (5.5) has parameters $[31, 6, 8]$ and its weight enumerator is given by $$1+2z^{8}+z^{10}+11z^{14}+45z^{16}+3z^{18}+z^{22}.$$  By Theorem 6.3 (2) and  $$\frac{w_{\min}}{w_{\max}}=\frac{8}{22}<1/2,$$ the code $\mathcal{C}_{f}$ is minimal violating the condition of Aschikhmin-Barg. 
\end{exa}

\section{Concluding remarks}

The main contributions of this paper are following

\begin{itemize}
\item Constructions of  binary linear codes defined in Eqs. (5.1) and (5.5)  associated with posets.

\item Determinations of  weight distributions of binary linear codes associated with hierarchical posets of two levels (Theorems 5.1, 5.3, 5.5, and 5.7).

\item Some infinite families of minimal and optimal binary linear codes violating Aschikhmin-Barg's condition (Theorems 6.1, 6.2, 6.3, and 6.4).
\end{itemize}

By huge computation,  more optimal binary linear codes can be found from Theorems 5.3, 5.5 and 5.7. As future works, it should be interesting to find optimal and minimal binary linear codes by using other posets. 

In this paper, we only considered linear codes generated by one or two order ideals in hierarchical posets of two levels. It also should be interesting to investigate the cases of more than two order ideals in hierarchical posets with two levels  or many levels.

\bigskip
\section*{Acknowledgments}
This research was supported by the National Natural Science Foundation of China (No. 61772015), the Foundation of Science and Technology on Information Assurance Laboratory (No. KJ-17-010),  the National Research Foundation of Korea grant funded by the Korea government (No. NRF-2017R1A2B2004574), and  the Basic Science Research Program through the National Research Foundation of Korea funded by the Ministry of Education (No. 2015049582). 

Part of this work was done when Yansheng Wu was visiting Korea Institute for Advanced Study (KIAS),  Seoul, South Korea. Yansheng Wu would like to thank the institution for the kind hospitality. 

The authors are very grateful to the reviewers and the Editor for their valuable comments and suggestions
to improve the quality of this  paper. The authors contribute 
equally to this paper.

\end{document}